\documentclass{article}
\usepackage{amsmath,amsthm,amssymb}
\newtheorem{Theorem}{Theorem}[section]

\newtheorem{Remark}[Theorem]{Remark}
\newtheorem{Definition}[Theorem]{Definition}
\newtheorem{Corollary}[Theorem]{Corollary}

\newtheorem{Example}[Theorem]{Example}
\numberwithin{equation}{section}

\begin{document}

\title{Constacyclic Codes over Finite Fields\footnote{
 E-Mail address: b\_c\_chen@yahoo.com.cn (B. C. Chen),
  yunfan02@yahoo.com.cn (Y. Fan),
  xiaomi1985516@126.com (L. R. Lin),
  h\_w\_liu@yahoo.com.cn (H. Liu).}}

\author{Bocong Chen,~ Yun Fan,~ Liren Lin,~ Hongwei Liu}

\date{\small School of Mathematics and Statistics,
Central China Normal University\\
Wuhan, Hubei, 430079, China}

\maketitle

\begin{abstract} An equivalence relation called isometry is introduced to classify constacyclic codes over a finite field;
the polynomial generators of constacyclic codes of length
$\ell^tp^s$ are characterized, where $p$ is the characteristic of
the finite field and $\ell$ is a prime different from $p$.

\medskip
\textbf{Keywords:} finite field, constacyclic code, isometry, polynomial generator.

\medskip
\textbf{2010 Mathematics Subject Classification:}~ 94B05; 94B15
\end{abstract}

\section{Introduction}

Constacyclic codes constitute a remarkable generalization of cyclic
codes, hence form an important class of linear codes in the coding
theory. And, constacyclic codes also have practical applications as
they can be encoded with shift registers.

In \cite{Bla}, for any positive integer $a$ and any odd integer $n$,
Blackford used the discrete Fourier transform to show that
${\bf Z}_4[X]/\langle X^{2^an}+1 \rangle$ is a principal ideal ring,
where ${\bf Z}_4$ denotes the residue ring of integers modulo $4$,
and to establish a concatenated structure of negacyclic codes
of length $2^an$ over ${\bf Z}_4$.
In \cite{Abu} Abualrub and Oehmke classified the cyclic codes
of length $2^k$ over ${\bf Z}_4$ by their generators.
Generalizing the result of \cite{Abu},
Dougherty and Ling in \cite{Dou} classified
the cyclic codes of length $2^k$ over
the Galois ring ${\rm GR}(4,m)$.

Let $F_q$ be a finite field with $q=p^m$ elements
where $p$ is a prime, and let $\lambda\in F_q^*$
where $F_q^*$ denotes the multiplicative group consisting of
all non-zero elements of~$F_q$.
Any $\lambda$-constacyclic code $C$ of length~$n$ over $F_q$
is identified with an ideal of the quotient algebra
$F_q[X]/\langle X^n-\lambda\rangle$ where
$\langle X^n-\lambda\rangle$ denotes the ideal generated
by $X^n-\lambda$ of the polynomial algebra $F_q[X]$,
hence $C$ is generated by a factor polynomial of $X^n-\lambda$,
called the {\em polynomial generator}
of the $\lambda$-constacyclic code $C$.
In order to obtain all $\lambda$-constacyclic codes
of length $n$ over~$F_q$, we need to determine
all the irreducible factors of $X^n-\lambda$ over $F_q$.
It is remarkable that, though all
irreducible binomials over $F_q$
have been explicitly characterized by Serret early in 1866
(e.g. see \cite[Theorem 3.75]{li} or \cite[Theorem 10.7]{Wan}),
no effective method were found to characterize
the irreducible factors of $X^n-\lambda$ over $F_q$ so far.
It is a challenge to determine explicitly the polynomial
generators of all constacyclic codes over finite fields.

It is well known that $X^n-\lambda$ is a factor of $X^N-1$ for
a suitable integer~$N$, and the irreducible factors
of $X^N-1$ over $F_q$ with $q=p^m$ as above
can be described by the $q$-cyclotomic cosets.
Recently, assuming that $p$ is odd and the order of $\lambda$
in the multiplicative group $F_q^*$ is a power of $2$,
Bakshi and Raka in \cite{Bak} described the polynomial generators
of $\lambda$-constacyclic codes of length $2^t$ over~$F_q$
by means of recognizing the $q$-cyclotomic cosets which are
corresponding to the irreducible factors of $X^{2^t}-\lambda$.
In the same paper \cite{Bak}, Bakshi and Raka determined
the polynomial generators of all the $\lambda$-constacyclic codes
of length $2^tp^s$ over $F_q$, $q = p^m$,
for any nonzero $\lambda$ in $F_q$.
Almost the same time but in another approach,
assuming that $p$ is odd,
Dinh in \cite{Dinh11} determined the polynomial generators
of all constacyclic codes of length $2p^s$ over $F_q$
in a very explicit form: the irreducible factors of the
polynomial generators are all binomials of degree $1$ or $2$.

In this paper, we are concerned with the constacyclic codes
of length $\ell^tp^s$ over $F_q$, where $q=p^m$ as before
and $\ell$ is a prime different from $p$.
We introduce a concept ``isometry'' for the non-zero elements
of $F_q$ to classify constacyclic codes over $F_q$ such that
the constacyclic codes belonging to the same isometry class
have the same distance structures and the same algebraic structures.
Then we characterize in an explicit way the polynomial generators
of constacyclic codes of length $\ell^tp^s$ over $F_q$
according to the isometry classes.
It is notable that, except for the
constacyclic codes which are isometric to cyclic codes,
the irreducible factors of the polynomial generator of
any constacyclic code of length $\ell^tp^s$ over $F_q$
are either all binomials or all trinomials.

The plan of this paper is as follows.
The necessary notations and some known results to be used
are provided in Section 2.
In Section 3, we introduce precisely the concept of isometry,
which is an equivalence relation on~$F_q^*$;
and some necessary and sufficient conditions for
any two elements of $F_q^*$ isometric to each other
are established; as a consequence, the constacyclic codes
isometric to cyclic codes are described.
In Section 4, we classify the constacyclic codes of
length $\ell^tp^s$ over $F_q$ into isometry classes,
characterize explicitly the polynomial generators of
the constacyclic codes of each isometry class, and derive some consequences, including the main result of \cite{Dinh11}.
In Section 5, with the help of the GAP (\cite{GAP}),
the polynomial generators of all constacyclic codes
of length 6 over $F_{2^4}$, all constacyclic codes
of length 175 over $F_{5^2}$ and all constacyclic codes
of length 20 over $F_{5^2}$ are computed.

\section{Preliminaries}
Throughout this paper $F_q$ denotes a finite field  with $q$
elements where $q=p^m$ is a power of a prime $p$.
Let $F_q^*$ denote the multiplicative group of $F_q$
consisting of all non-zero elements of $F_q$; and
for $\beta\in F_q^*$, let $\rm{ord}(\beta)$
denote the order of $\beta$ in the group $F_q^*$;
then $\rm{ord}(\beta)$ is a divisor of $q-1$, and $\beta$
is called a {\em primitive $\rm{ord}(\beta)$-th root of unity}.
It is well-known that $F_q^*$ is a cyclic group of order $q-1$,
i.e. $F_q^*$ is generated by a primitive $(q-1)$-th root
$\xi$ of unity, we denote it by $F_q^*=\langle\xi\rangle$.
For any integer $k$, it is known that
${\rm ord}(\xi^k)=\frac{q-1}{\gcd(k,q-1)}$,
where $\gcd(k,q-1)$ denotes the greatest common divisor
of $k$ and $q-1$.

Assume that $n$ is a positive integer and
$\lambda$ is a non-zero element of $F_q$.
A linear code $C$ of length $n$ over $F_q$ is said to be
{\it $\lambda$-constacyclic} if for any code word
$(c_0,c_1,\cdots, c_{n-1})\in C$ we have that $(\lambda c_{n-1},c_0,c_1,\cdots, c_{n-2})\in C$. We denote by
$F_q[X]$, the polynomial algebra over $F_q$, and denote by
$\langle X^n-\lambda \rangle$, the ideal of $F_q[X]$
generated by $X^n-\lambda$. Any element of the quotient algebra
$F_q[X]/\langle X^n-\lambda\rangle$ is uniquely represented by
a polynomial $a_0+a_1X+\cdots+a_{n-1}X^{n-1}$ of degree less than $n$,
hence is identified with a word $(a_0,a_1,\cdots,a_{n-1})$
of length $n$ over $F_q$; so we have the corresponding
{\em Hamming weight} and the {\em Hamming distance}
on the algebra $F_q[X]/\langle X^n-\lambda\rangle$.

In this way, any $\lambda$-constacyclic code $C$
of length $n$ over $F_q$ is identified with exactly one ideal
of the quotient algebra $F_q[X]/\langle X^n-\lambda\rangle$,
which is generated by a divisor $g(X)$ of $X^n-\lambda$,
and the divisor $g(X)$ is determined by $C$ uniquely up to a scale;
in that case, $g(X)$ is called a {\em polynomial generator}
of $C$ and write it as $C=\langle g(X)\rangle$.
Specifically, the irreducible factorization of $X^n-\lambda$
in $F_q[X]$ determines all $\lambda$-constacyclic
codes of length $n$ over $F_q$.

Note that the $1$-constacyclic codes are just the usual
{\em cyclic codes}, and there is a lot of literature
to deal with the cyclic codes.
In particular, the irreducible factorization of
$X^n-1$ in $F_q[X]$ can be described as follows.
As usual, we adopt the notations:
$k\,|\,n$ means that the integer $k$ divides $n$;
and, for a prime integer~$\ell$,
$\ell^e\Vert n$ means that $\ell^e\,|\,n$ but $\ell^{e+1}\nmid n$.

\begin{Remark}\label{irr-cyclic}\rm
Assume that $n=n'p^s$ with $s\ge 0$ and $p\nmid n'$.
For an integer $r$ with $0\leq r\leq n'-1$,
the {\em $q$-cyclotomic coset of $r$ modulo $n'$} is defined by
$$C_r=\{r\cdot q^j~(\bmod~n')\,|\,j=0,1,\cdots\}.$$
A subset $\{r_1, r_2, \cdots ,r_\rho\}$ of $\{0,1,\cdots ,n'-1\}$ is called a {\em complete set of representatives}
of all $q$-cyclotomic cosets modulo $n'$ if
$C_{r_1}, C_{r_2}, \cdots ,C_{r_\rho}$ are distinct and
$\bigcup_{i=1}^\rho C_{r_i}=\{0,1,\cdots ,n'-1\}$.
Take $\eta$ to be a primitive $n'$-th root of unity
(maybe in an extension of $F_q$), and denote by $M_{\eta}(X)$,
the minimal polynomial of $\eta$ over $F_q$.
It is well-known that (e.g. see \cite[Theorem 4.1.1]{huff}):
\begin{equation}\label{simple-irr-decomposition}
X^{n'}-1= M_{\eta^{r_1}}(X)M_{\eta^{r_2}}(X)
 \cdots M_{\eta^{r_\rho}}(X)
\end{equation}
with
$$M_{\eta^{r_i}}(X)=\prod\limits_{j\in C_{r_i}}(X-\eta^j),
 \qquad i=1,\cdots,\rho,$$
all being irreducible in $F_q[X]$, hence
\begin{equation}\label{irr-decomposition}
X^n-1=(X^{n'}-1)^{p^s}= M_{\eta^{r_1}}(X)^{p^s}
M_{\eta^{r_2}}(X)^{p^s}\cdots M_{\eta^{r_\rho}}(X)^{p^s}
\end{equation}
is the irreducible decomposition of $X^n-1$ in $F_q[X]$.
\end{Remark}

\smallskip
In a very special case the irreducible factorization
of $X^n-\lambda$ in $F_q[X]$ has been characterized precisely,
we quote it as the following remark.

\begin{Remark}\label{irr-trinomial}\rm Assume that
$q\equiv3~(\bmod~4)$ (in particular, $q$ is a power of an odd prime), equivalently, $2\Vert(q-1)$.
Then $X^{2^t}+1$ is factorized into irreducible polynomials
over $F_q$ in \cite[Theorem 1]{Blake}. We should mention that,
though \cite[Theorem 1]{Blake} is proved for a prime $p$
with $p\equiv3~(\bmod~4)$,
one can check in the same way as in \cite{Blake}
that it also holds for the present case when
$q$ is a power of a prime and $q\equiv3~(\bmod~4)$.
We reformulate the result as follows.
Note that $4\,|\,(q+1)$ in the present case,
hence there is an integer $e\ge 2$ such that $2^e\Vert(q+1)$.
Set $H_1=\{0\}$; recursively define
$$\textstyle H_i=\left\{\pm(\frac{h+1}{2})^\frac{q+1}{4}\,\big|\,h\in H_{i-1}\right\},$$
for $i=2,3, \cdots, e-1$; and set
$$\textstyle H_e=\left\{\pm(\frac{h-1}{2})^\frac{q+1}{4}\,\big|\,h\in H_{e-1}\right\}
   =H_{e+1}=H_{e+2}=\cdots.$$
Let $t\ge 1$. Set $b=t$ and $c=0$ if $1\leq t\leq e-1$;
while set $b=e$ and $c=1$ if $t\ge e$.
Then (see \cite[Theorem 1]{Blake} or \cite[Theorem 10.13]{Wan}):
\begin{equation}\label{irr-2-decomposition}
 X^{2^t}+1=\prod\limits_{h\in H_{t}}
 \big(X^{2^{t-b+1}}-2hX^{2^{t-b}}+(-1)^c\big)
\end{equation}
with all the factors in the right hand side being
irreducible over $F_{q}$.
\end{Remark}

\smallskip Return to our general case.
As we mentioned before,
the irreducible non-linear binomials over $F_q$
have been determined by Serret early in 1866
(see \cite[Theorem 3.75]{li} or \cite[Theorem 10.7]{Wan}), we restate it as
a remark for later quotations.

\begin{Remark}\label{irr-binomial}\rm
Assume that $n\ge 2$. For any $a\in F_q^*$ with ${\rm ord}(a)=k$,
the binomial $X^n-a$ is irreducible over $F_q$ if and only if
both the following two conditions are satisfied:

{\rm(i)}\quad
Every prime divisor of $n$ divides $k$,
but does not divide  $(q-1)/k$;

{\rm(ii)}\quad If $4\,|\,n$, then $4\,|\,(q-1)$.
\end{Remark}

\section{Isometries between constacyclic codes}

Let $F_q$ be a finite field of order $q=p^m$
and $F_q^*=\langle\xi\rangle$ as before, where $\xi$ is a
primitive $(q-1)$-th root of unity.
Let $n$ be a positive integer.

Generalizing the usual equivalence between codes,
we consider a kind of equivalences between the
$\lambda$-constacyclic codes and the $\mu$-constacyclic codes
which preserve the algebraic structures of the constacyclic codes.

\begin{Definition}\label{n-equivalence}
Let $\lambda, \mu\in F_q^*$. We say that
an $F_q$-algebra isomorphism
$$
 \varphi:\quad F_q[X]/\langle X^n-\mu\rangle~\longrightarrow~
 F_q[X]/\langle X^n-\lambda\rangle
$$
is an {\em isometry} if it preserves the Hamming distances
on the algebras, i.e.
$$
 d_H\big(\varphi({\bf a}),\varphi({\bf a}')\big)=
 d_H({\bf a},{\bf a}'),\qquad\forall~
 {\bf a},{\bf a}'\in F_q[X]/\langle X^n-\mu\rangle.
$$
And, if there is an isometry between
$F_q[X]/\langle X^n-\lambda\rangle$
and $F_q[X]/\langle X^n-\mu\rangle$, then
we say that $\lambda$ is {\em $n$-isometric}
to $\mu$ in $F_q$, and denote it $\lambda\cong_n\mu$.
\end{Definition}

Obviously, the $n$-isometry ``$\cong_n$''
is an equivalence relation on~$F_q^*$, hence $F_q^*$ is partitioned
into $n$-isometry classes.
If $\lambda\cong_n\mu$, then all the $\lambda$-constacyclic codes
of length $n$ are one to one corresponding to all the
$\mu$-constacyclic codes of length $n$ such that
the corresponding constacyclic codes have the same dimension
and the same distance distribution,
specifically, have the same minimum distance; at that case
we say that, for convenience,
the $\lambda$-constacyclic codes of length $n$
are {\em isometric} to the $\mu$-constacyclic codes of length~$n$.
So, it is enough to study the $n$-isometry classes of constacyclic codes.

\begin{Theorem}\label{n-isometry} For any $\lambda,\mu\in F_q^*$,
the following three statements are equivalent to each other:

{\bf(i)}\quad $\lambda\cong_n\mu$.

{{\bf(ii)}\quad $\langle\lambda,\xi^n\rangle=\langle\mu,\xi^n\rangle$,
where $\langle\lambda,\xi^n\rangle$ denotes the subgroup of $F_q^*$
generated by $\lambda$ and $\xi^n$.

\bf(iii)}\quad There is a positive integer $k<n$ with
$\gcd(k,n)=1$ and an element $a\in F_q^*$ such that
$a^n\lambda=\mu^k$ and the following map
\begin{equation}\label{a-isometry}
\varphi_{a}:\quad F_q[X]/\langle X^n-\mu^k\rangle~\longrightarrow~
 F_q[X]/\langle X^n-\lambda\rangle,
\end{equation}
which maps any element $f(X)+\langle X^n-\mu^k\rangle$ of
$F_q[X]/\langle X^n-\mu^k\rangle$ to
the element $f(aX)+\langle X^n-\lambda\rangle$
of $F_q[X]/ \langle X^n-\lambda\rangle$,
is an isometry.

\noindent
In particular, the number of $n$-isometry classes of $F_q^*$
is equal to the number of positive divisors of $\gcd(n, q-1)$.
\end{Theorem}

\begin{proof}  (i) $\Rightarrow$ (ii).~ By (i) we have
an isometry $\varphi$ between the algebras:
$$\varphi:\quad F_q[X]/\langle X^n-\mu\rangle~
 \longrightarrow~F_q[X]/\langle X^n-\lambda\rangle.$$
Since $\varphi$ preserves the Hamming distance,
it must map $X$ of weight~$1$ of the algebra
$F_q[X]/\langle X^n-\mu\rangle$ to an element of the algebra
$F_q[X]/\langle X^n-\lambda\rangle$ of weight~$1$,
so there is an element $b\in F_q^*$ and an integer
$j$ with $0\le j<n$ such that
\begin{equation}\label{j-iso}
 \varphi(X)=bX^j\,.
\end{equation}
Consider $\varphi(X^i)=(bX^j)^i=b^iX^{ji}~(\bmod~X^n-\lambda)$~
for $i=0,1,\cdots,n-1$;
since $\varphi$ is a bijection,
we see that any index $e$ with $0\le e\le n-1$ must appear
in the following sequence:
$$
 ji~(\bmod~n),\qquad i=0,1,\cdots,n-1;
$$
hence $j~(\bmod~n)$ must be invertible,
i.e. $0<j<n$ and $\gcd(j,n)=1$.
Note that $X^n=\lambda~(\bmod~X^n-\lambda)$; further,
note that $\varphi$ is an algebra isomorphism
and $\mu\in F_q$, we see that $\varphi(\mu)=\mu$,
and can make the following calculation in
$F_q[X]/\langle X^n-\lambda\rangle$
(or equivalently, modulo $X^n-\lambda$):
\begin{equation}\label{eq2}
\mu=\varphi(\mu)=\varphi(X^n)=\varphi(X)^n
=(bX^j)^n=b^nX^{jn}=b^n\lambda^j;
\end{equation}
i.e. as elements of $F_q$ we have $\mu=\lambda^jb^n$.
Obviously, $\langle\xi^n\rangle=\{a^n\mid a\in F_q^*\}$.
We have $\mu\in\langle\lambda,\xi^n\rangle$, and hence
$\langle\mu,\xi^n\rangle\subseteq\langle\lambda,\xi^n\rangle$.
On the other hand, since $\gcd(j,n)=1$,
there are integers $k,h$ such that $jk+nh=1$; so
$$
\mu^{k}=\lambda^{jk}b^{nk}=\lambda^{jk+nh}\lambda^{-nh}b^{nk}
=\lambda(\lambda^{-h}b^{k})^n;
$$
i.e.
$\lambda=\mu^{k}(\lambda^{h}b^{-k})^n\in\langle\mu,\xi^n\rangle$;
and we have that
$\langle\lambda,\xi^n\rangle\subseteq\langle\mu,\xi^n\rangle$.
Thus, we get the desired conclusion:
$\langle\lambda,\xi^n\rangle=\langle\mu,\xi^n\rangle$.

\smallskip(ii) $\Rightarrow$ (iii).~ Denote $d=\gcd(n,q-1)$.
Then the subgroup $\langle\xi^n\rangle=\langle\xi^d\rangle$,
and the quotient group
$$F_q^*/\langle\xi^n\rangle=F_q^*/\langle\xi^d\rangle
  =\langle\xi\rangle/\langle\xi^d\rangle$$
is a cyclic group of order $d$.
From the statement (ii) we have that
$$\langle\lambda,\xi^n\rangle/\langle\xi^d\rangle
  =\langle\mu,\xi^n\rangle/\langle\xi^d\rangle;$$
which implies that, in the cyclic group
$F_q^*/\langle\xi^d\rangle$ of order $d$, $\lambda$ and $\mu$ generate
the one and the same subgroup, in particular,
they have the same order in the quotient group $F_q^*/\langle\xi^d\rangle$.
Thus there are integers $k',h'$ such that
$\lambda=\mu^{k'}\xi^{dh'}$ and $\gcd(k',d)=1$.
Since $d\mid n$, it is known that the natural map
$$
{\bf Z}_n^*~\longrightarrow~{\bf Z}_d^*,\quad
 z~(\bmod~n)~\longmapsto~z~(\bmod~d),
$$
is a surjective homomorphism,
where ${\bf Z}_n^*$ denotes the multiplicative group
consisting of all reduced residue classes modulo $n$.
We can take a positive integer $k<n$ with $\gcd(k,n)=1$ and
$k\equiv k'~(\bmod~d)$. Then there is an integer~$h$ such that
$k'=k+dh$. So
$$
\lambda=\mu^{k'}\xi^{dh'}=\mu^{k+dh}\xi^{dh'}=\mu^k(\mu^h\xi^{h'})^d.
$$
As $(\mu^h\xi^{h'})^d\in\langle\xi^d\rangle
=\langle\xi^n\rangle$, we have an $a\in F_q^*$ such that
$(\mu^h\xi^{h'})^d=a^{-n}$.
In a word, we have an integer $k$ coprime
to $n$ and an $a\in F_q^*$ such that $a^n\lambda=\mu^k$.
Now we define an algebra homomorphism:
$$
\hat\varphi_{a}:\quad F_q[X]~\longrightarrow~
 F_q[X]/\langle X^n-\lambda\rangle,
$$
by mapping $f(X)\in F_q[X]$ to
$\hat\varphi_a\big(f(X)\big)=f(aX)~(\bmod~X^n-\lambda)$;
since $a$ is non-zero, $\hat\varphi_{a}$ is obviously surjective.
Noting that $X^n=\lambda~(\bmod~X^n-\lambda)$, we have
$$
\hat\varphi_a(X^n-\mu^k)=(aX)^n-\mu^k=a^n X^n-\mu^k
  = a^n\lambda-\mu^k = 0 \pmod{X^n-\lambda}.
$$
So the surjective algebra homomorphism $\hat\varphi_a$
induces an algebra isomorphism
\begin{equation*}
\varphi_{a}:\quad F_q[X]/\langle X^n-\mu^k\rangle~\longrightarrow~
 F_q[X]/\langle X^n-\lambda\rangle,
\end{equation*}
which maps any element $f(X)+\langle X^n-\mu^k\rangle$ of
$F_q[X]/\langle X^n-\mu^k\rangle$ to
the element $f(aX)+\langle X^n-\lambda\rangle$
of $F_q[X]/ \langle X^n-\lambda\rangle$;
since $\varphi_a$ maps any element $X^i$ of weight $1$
to an element $a^iX^i$ of weight $1$,
the algebra isomorphism $\varphi_a$ preserves Hamming distances
of the algebras. we are done for the statement (iii).

\smallskip{(iii)} $\Rightarrow$ {(i)}.~
Since the map (\ref{a-isometry}) in the statement (iii)
is an algebra isomorphism,
we have that
$$0=\varphi_{a}(X^n-\mu^k)=(aX)^n-\mu^k
 =a^n\lambda-\mu^k \pmod{X^n-\lambda};$$
that is, $\lambda a^n=\mu^k$.
By (iii) it is assumed that $\gcd(k,n)=1$, i.e.
there are integers $j,h$ such that $kj+nh=1$,
which also implies that $\gcd(j,n)=1$; so
$$
\mu=\mu^{kj+nh}=(\mu^{k})^{j}\mu^{nh}
  =(\lambda a^n)^{j}\mu^{hn}=\lambda^j(a^j\mu^h)^n.
$$
Set $b=a^j\mu^h$, then $b\in F_q^*$ and $b^n\lambda^j=\mu$.
Since $F_q[X]$ is a free $F_q$-algebra with $X$ as a free generator,
by mapping $X$ to $bX^j$, we can define an algebra homomorphism:
$$
\hat\varphi:\quad F_q[X]~\longrightarrow~
 F_q[X]/\langle X^n-\lambda\rangle,
$$
which maps any $f(X)\in F_q[X]$ to
$\hat\varphi\big(f(X)\big)=f(bX^j)~(\bmod~X^n-\lambda)$.
Since $j$ is coprime to $n$, the following
$$
 \hat\varphi(X^i)=b^iX^{ji}\pmod{X^n-\lambda},\qquad
 i=0,1,\cdots,n-1,
$$
form a basis of the algebra $F_q[X]/\langle X^n-\lambda\rangle$;
so $\hat\varphi$ is a surjective algebra homomorphism.
Further, we have
$$
\hat\varphi(X^n-\mu)=(bX^j)^n-\mu=b^n X^{nj}-\mu
  = b^n\lambda^j-\mu = 0 \pmod{X^n-\lambda}.
$$
Thus the surjective algebra homomorphism $\hat\varphi$
induces an algebra isomorphism:
\begin{equation*}
\varphi:\quad F_q[X]/\langle X^n-\mu\rangle~\longrightarrow~
 F_q[X]/\langle X^n-\lambda\rangle,
\end{equation*}
which maps any element $f(X)+\langle X^n-\mu\rangle$ of
$F_q[X]/\langle X^n-\mu\rangle$ to
the element $f(bX^j)+\langle X^n-\lambda\rangle$
of $F_q[X]/ \langle X^n-\lambda\rangle$;
in particular, $\varphi$ maps any element $X^i$ of weight $1$
to an element $b^iX^{ji}~(\bmod~X^n-\lambda)$ of weight $1$,
hence $\varphi$ preserves the Hamming distances.
That is, (i) holds.

\smallskip
Finally, by the equivalence of (i) and (ii),
the number of the $n$-isometry classes of $F_q^*$
is equal to the number of the subgroups of the
quotient group $F_q^*/\langle\xi^d\rangle$
where $d=\gcd(n,q-1)$. The quotient $F_q^*/\langle\xi^d\rangle$
is a cyclic group of order~$d$, so, for any divisor $d'\,|\,d$
it has a unique subgroup of order $d'$.
Then the number of the subgroups of $F_q^*/\langle\xi^d\rangle$
is equal to the number of the positive divisors of~$d$.
In conclusion, the number of the $n$-isometry classes of $F_q^*$
is equal to the number of the positive divisors of $\gcd(n,q-1)$.
\end{proof}

\begin{Remark}\label{rem-on-isometry}\rm
Though the statement (i) of Theorem \ref{n-isometry} states that
there is an isometry
$\varphi:F_q[X]/\langle X^n-\mu\rangle\to
F_q[X]/\langle X^n-\lambda\rangle$,
the statement (iii) of Theorem~\ref{n-isometry} exhibits
a specific isometry $\varphi_a$ such that $\varphi_a(X)=aX$,
which outperforms $\varphi$ in (\ref{j-iso}) and provides
an easy way to connect the polynomial generators of
the $\lambda$-constacyclic codes with those of the
$\mu^k$-constacyclic codes.
\end{Remark}

In particular, taking $\mu=1$, we see that
$\lambda\cong_n 1$ implies that there is an isometry
$\varphi_a: F_q[X]/\langle X^n-1\rangle\to
F_q[X]/\langle X^n-\lambda\rangle$ such that $\varphi(X)=aX$.
Thus for the constacyclic codes $n$-isometric
to the cyclic codes, we have the following consequence
which is closely related to \cite[Lemma 3.1]{Hug}.

\begin{Corollary}\label{thm-cyclic}
Let $n$ be a positive integer, and $\lambda\in F_q^*$.
The $\lambda$-constacyclic codes of length $n$
are isometric to the cyclic codes of length $n$ if and only if
$a^n\lambda=1$ for an element $a\in F_q^*$; further,
in that case the map
\begin{equation}\label{cyclic-isometry}
\varphi_{a}:\quad F_q[X]/\langle X^n-1\rangle~\longrightarrow~
  F_q[X]/\langle X^n-\lambda \rangle,
\end{equation}
which maps $f(X)$ to $f(aX)$, is an isometry, and
\begin{equation}\label{lambda-irr-decomp}
X^n-\lambda = \lambda\cdot M_{\eta^{r_1}}(aX)^{p^s}
 M_{\eta^{r_2}}(aX)^{p^s}\cdots M_{\eta^{r_\rho}}(aX)^{p^s}
\end{equation}
is an irreducible factorization of $X^n-\lambda$ in $F_q[X]$,
where $n=n'p^s$ with $s\ge 0$ and $p\nmid n'$,
$M_{\eta^i}(X)$ and $\{r_1,\cdots,r_\rho\}$
are defined in the formula (\ref{irr-decomposition});
in particular, any $\lambda$-constacyclic code $C$
has a polynomial generator as follows:
\begin{equation}\label{cyclic-generator}
\prod_{i=1}^{\rho} M_{\eta^{r_i}}(aX)^{e_i},\qquad
 0\le e_i\le p^s,~\forall~i=1,\cdots,\rho.
\end{equation}
\end{Corollary}

\begin{proof}
By Theorem \ref{n-isometry}, $\lambda\cong_n 1$ if and only if
$\langle\lambda,\xi^n\rangle=\langle 1,\xi^n\rangle
=\langle\xi^n\rangle$; i.e.
$\lambda\cong_n 1$ if and only if $\lambda\in\langle\xi^n\rangle$.
However, $\langle\xi^n\rangle=\{a^n\mid a\in F_q^*\}$;
so $\lambda\cong_n 1$ if and only if
$\lambda=b^n$ for an element $b\in F_q^*$.

Assume that it is the case, i.e. $a^{n}\lambda=1$.
By the statement (iii) of Theorem \ref{n-isometry},
the map (\ref{cyclic-isometry}) is an isometry between the algebras.
And, as in the formula~(\ref{irr-decomposition}),
we have the irreducible decomposition of $X^n-1$ in $F_q[X]$:
$$X^n-1 = M_{\eta^{r_1}}(X)^{p^s}M_{\eta^{r_2}}(X)^{p^s}
 \cdots M_{\eta^{r_\rho}}(X)^{p^s};$$
hence the following is an irreducible decomposition
of $(aX)^n-1$ in $F_q[X]$:
$$(aX)^n-1 = M_{\eta^{r_1}}(aX)^{p^s}M_{\eta^{r_2}}(aX)^{p^s}
 \cdots M_{\eta^{r_\rho}}(aX)^{p^s}.$$
However, since $a^{n}=\lambda^{-1}$,
we have that $(aX)^n=a^{n}X^n=\lambda^{-1}X^n$;
thus we get the irreducible decomposition of
$X^n-\lambda$ in $F_q[X]$ in the formula (\ref{lambda-irr-decomp}).
Finally, the polynomial generator of any $\lambda$-constacyclic
code is a divisor of $X^n-\lambda$, hence
has the form in (\ref{cyclic-generator}).
\end{proof}

\begin{Corollary}\label{cor3-1}
If $n$ is a positive integer coprime to $q-1$, then
there is only one $n$-isometry class in $F_q^*$; in particular,
for any $\lambda\in F_q^*$
the $\lambda$-constacyclic codes of length~$n$
are isometric to the cyclic codes of length~$n$,
i.e. $a^n\lambda=1$ for an $a\in F_q^*$ and all the
(\ref{cyclic-isometry}), (\ref{lambda-irr-decomp}) and
(\ref{cyclic-generator}) hold.
\end{Corollary}

\begin{proof}
Since $\gcd(n,q-1)=1$, the conclusion is obtained immediately.
It is an automorphism of the group $F_q^*$
which maps any $a\in F_q^*$ to $a^n\in F_q^*$;
thus there is a $b\in F_q^*$ such that $\lambda=b^n$.
\end{proof}

Let $n=n'p^s$ as in Corollary \ref{thm-cyclic}.
If $n'=1$, then $n=p^s$ is coprime to $q-1$ and
$X^{p^s}-1=(X-1)^{p^s}$, and we get the following result
at once.

\begin{Corollary}\label{cor4-2}
For any $\lambda\in F_q^*$ the $\lambda$-constacyclic codes
of length $p^s$ are isometric to the cyclic codes of length $p^s$;
in particular, there is an $a\in F_q^*$ such that
$a^{p^s}\lambda=1$ and $X^{p^s}-\lambda=\lambda(aX-1)^{p^s}$ is an
irreducible factorization in $F_q[X]$;
in particular, any $\lambda$-constacyclic code $C$ of length $p^s$
has a polynomial generator $(X-a^{-1})^i$ with $0\le i\le p^s$.
\qed
\end{Corollary}

\begin{Remark}\label{rem1}\rm
Taking $\lambda=-1$, Corollary \ref{cor4-2} implies that
negacyclic codes of length $p^s$ are isometric to
cyclic codes of length $p^s$.
This generalizes \cite[Theorem 3.3]{Dinh10}
which showed that, in our terminology,
$\lambda$-constacyclic codes of length $p^s$
over $F_{p^m}$ are isometric to the negacyclic codes of
length $p^s$ over $F_{p^m}$.

\end{Remark}

\section{Constacyclic codes of length $\ell^tp^s$}

Let $F_q$ be a finite field of order $q=p^m$ and
$F_q^*=\langle\xi\rangle$ be generated by a primitive
$(q-1)$-th root $\xi$ of unity as before.

In this section, we consider constacyclic codes
of length $\ell^tp^s$ over $F_q$,
where~$\ell$ is a prime integer different from $p$ and
$s$, $t$ are non-negative integers.
We will show that
any $\lambda$-constacyclic code of length $\ell^tp^s$
with $\lambda\not\cong_{\ell^tp^s}1$
has a polynomial generator with irreducible factors all being
binomials of degrees equal to powers of the prime $\ell$
except for the case when $\ell=2$, $t\ge 2$ and $2\Vert(q-1)$;
and in the exceptional case the polynomial generator
with irreducible factors all being trinomials
corresponding to the factorization (\ref{irr-2-decomposition}).

As we did in Remark \ref{irr-cyclic},
take a complete set $\{r_1,\cdots,r_\rho\}$ of representatives
of $q$-cyclotomic cosets modulo $\ell^t$;
take a primitive $\ell^t$-th root $\eta$ of unity
(maybe in an extension of $F_q$), and denote $M_{\eta}(X)$
the minimal polynomial of $\eta$ over $F_q$;
by the formula (\ref{irr-decomposition}),
\begin{equation}\label{lp-irr-decomp}
X^{\ell^tp^s}-1=(X^{\ell^t}-1)^{p^s}= M_{\eta^{r_1}}(X)^{p^s}
M_{\eta^{r_2}}(X)^{p^s}\cdots M_{\eta^{r_\rho}}(X)^{p^s}
\end{equation}
is the irreducible factorization of $X^{\ell^tp^s}-1$ in $F_q[X]$.
Further, assume that
\begin{equation}\label{u-zeta-v}
\ell^u\Vert(q-1)\,,\qquad\zeta=\xi^{\frac{q-1}{\ell^u}}\,,\qquad v=\min\{t,u\}.
\end{equation}

\begin{Theorem}\label{thm4-1}
With notations as above, for any $\lambda\in F_q^*$
there is an index~$j$ with $0\le j\le v$ such that
$\lambda\cong_{\ell^tp^s}\zeta^{\ell^j}$
and one of the following two cases holds:
\begin{itemize}
\item[{\bf(i)}]  $j=v$, then $\lambda\cong_{\ell^tp^s}1$,
$a^{\ell^tp^s}\lambda=1$ for an $a\in F_q^*$ and
$X^{\ell^tp^s}-\lambda=\lambda\cdot
 \prod_{i=1}^{\rho}M_{\eta^{r_i}}(aX)^{p^s}$
with $\{r_1,\cdots,r_\rho\}$ and $M_{\eta^{r_i}}(X)$'s
defined in (\ref{lp-irr-decomp}).
\item[{\bf(ii)}] $0\le j\le v-1$, then
$a^{\ell^tp^s}\lambda=\zeta^{k\ell^j}$ for an $a\in F_q^*$
and a positive integer~$k$ coprime to $\ell^tp^s$;
there are two subcases:
\begin{itemize}
\item[{\bf(ii.a)}]
 if $\ell=2$, $t\geq 2$ and $2\Vert(q-1)$,
 then $j=0$, $a^{\ell^tp^s}\lambda=-1$
 and, setting $H_t$, $b$ and $c$ to be as in
 Remark~\ref{irr-trinomial}, we have that
\begin{equation}\label{lambda-2p-decomp}
X^{2^tp^s}-\lambda=(-\lambda)\cdot\prod\limits_{h\in H_{t}}
\big(a^{2^{t-b+1}}X^{2^{t-b+1}}-2a^{2^{t-b}}hX^{2^{t-b}}+(-1)^c\big)^{p^s}
\end{equation}
with all the factors in the right hand side being
irreducible over $F_{q}$;
\item[{\bf(ii.b)}]
otherwise, taking an integer $s'$ with $0\le s'<m$
and $s'\equiv s~(\bmod~m)$, we have that
\begin{equation}\label{lambda-lp-decomp}
X^{\ell^tp^s}-\lambda=\prod_{i=0}^{\ell^j-1}
 \left(X^{\ell^{t-j}}-a^{-\ell^{t-j}}\zeta^{i\ell^{u-j}+kp^{m-s'}}\right)^{p^s}
\end{equation}
with all the factors in the right hand side being
irreducible over $F_{q}$.
\end{itemize}\end{itemize}
\end{Theorem}

\begin{proof}~ As $q-1=p^m-1$, it is clear that $\gcd(p^s, q-1)=1$.
From the notation~(\ref{u-zeta-v}), we have

$\bullet$~
$\zeta\in F_q$ is a primitive $\ell^u$-th root of unity,
$\langle\zeta\rangle$ is the Sylow $\ell$-subgroup of $F_q^*$,
and $\zeta^{\ell^{u-j}}$ for $0\le j\le u$ is a primitive $\ell^{j}$-th root of unity;

$\bullet$~ $\ell^v=\gcd(\ell^tp^s,\,q-1)$, so
${\rm ord}(\xi^{\ell^tp^s})=\frac{q-1}{\gcd(\ell^t p^s,q-1)}
=\frac{q-1}{\ell^v}={\rm ord}(\xi^{\ell^v})$, hence in the multiplicative
group $F_q^*$ we have that
\begin{equation}\label{subg}
\langle\xi^{\ell^tp^s}\rangle=\langle\xi^{\ell^t}\rangle
 =\langle\xi^{\ell^v}\rangle
\end{equation}
which is a subgroup of $F_q^*$ of order $\frac{q-1}{\ell^v}$.

Thus the quotient group $F_q^*/\langle\xi^{\ell^v}\rangle$
is a cyclic group of order $\ell^v$; and
for each positive divisor $\ell^{v-j}$ of $\ell^v$,
where $j=0,1,\cdots,v$, $\langle\zeta^{\ell^j},\xi^{\ell^v}\rangle
 /\langle\xi^{\ell^v}\rangle$
is the unique subgroup of order $\ell^{v-j}$
of the quotient group $F_q^*/\langle\xi^{\ell^v}\rangle$.

By the equivalence (i)$\Leftrightarrow$(ii) of Theorem
\ref{n-isometry}, the number of the $\ell^tp^s$-isometry classes
of $F_q^*$ is equal to $v+1$; precisely,
for any $\lambda\in F_q^*$ there is exactly one index $j$
with $0\le j\le v$ such that
$\lambda\cong_{\ell^tp^s}\zeta^{\ell^j}$.
We continue the discussion in two cases.

\smallskip{\em Case (i)}:\quad
$j=v$, i.e. $\lambda\cong_{\ell^tp^s}\zeta^{\ell^v}$;
by the equality (\ref{subg}), we see that
$\langle\lambda,\xi^{\ell^tp^s}\rangle
  =\langle \zeta^{\ell^v},\xi^{\ell^v}\rangle=\langle 1,\xi^{\ell^tp^s}\rangle$,
in other words, $\lambda\cong_{\ell^tp^s}1$.
By Corollary~\ref{thm-cyclic},
$a^{\ell^tp^s}\lambda=1^k=1$ for an $a\in F_q^*$, and
from the irreducible factorization~(\ref{lp-irr-decomp})
we get the irreducible factorization
$X^{\ell^tp^s}-\lambda=\lambda\cdot
 \prod_{i=1}^{\rho}M_{\eta^{r_i}}(aX)^{p^s}$.

\medskip{\em Case (ii)}:~ $0\le j\le v-1$.
Then by (\ref{u-zeta-v}) we have
\begin{equation}\label{case-ii}
 0\le j\le v-1<v=\min\{t,u\}\,,
\end{equation}
in particular, $v\ge 1$, i.e. $\ell\,|\,(q-1)$;
further, since $\lambda\cong_{\ell^tp^s}\zeta^{\ell^j}$,
by Theorem~\ref{n-isometry}~(iii) there is an $a\in F_q^*$
and a positive integer $k$ such that
\begin{equation}\label{a-k}
 a^{\ell^tp^s}\lambda=\zeta^{k\ell^j}\,,\qquad\gcd(k,\,\ell^tp^s)=1.
\end{equation}
We discuss it in the two subcases (ii.a) and (ii.b) as described in
the theorem.

{\rm{\em Subcase (ii.a)}}.~
Since $\ell=2$, $t\geq2$ and $2\Vert(q-1)$,
we have that $q$ is odd, $t>u=v=1$, $\zeta=-1$ and $j=0$;
and, from (\ref{a-k}) we see that $\ell=2\nmid k$ and
$a^{2^tp^s}\lambda=(-1)^k=-1$.
From the formula~(\ref{irr-2-decomposition}),
we have the following irreducible factorization in $F_q[X]$:
$$
 X^{2^tp^s}+1=\prod\limits_{h\in H_{t}}
 \big(X^{2^{t-b+1}}-2hX^{2^{t-b}}+(-1)^c\big)^{p^s};
$$
thus the following is an irreducible factorization
of $(aX)^{2^tp^s}+1$ in $F_q[X]$:
$$(aX)^{2^tp^s}+1 = \prod\limits_{h\in H_{t}}
\big(a^{2^{t-b+1}}X^{2^{t-b+1}}-2a^{2^{t-b}}hX^{2^{t-b}}+(-1)^c\big)^{p^s}.$$
However, since $a^{2^tp^s}=-\lambda^{-1}$,
we have that $(aX)^{2^tp^s}=a^{2^tp^s}X^{2^tp^s}=-\lambda^{-1}X^{2^tp^s}$;
thus we get the irreducible factorization (\ref{lambda-2p-decomp})
of $X^{2^tp^s}-\lambda$ in $F_q[X]$.

{\rm{\em Subcase (ii.b)}}.~
Remember that the conclusion in Remark \ref{irr-binomial}
is applied in this subcase.

By the choice of $s'$, $m-s'+s\equiv 0~(\bmod~m)$, so
$(p^m-1)\mid(p^{m-s'+s}-1)$, i.e.
$p^{m-s'+s}\equiv 1~(\bmod~q-1)$;
in particular, $\beta^{p^{m-s'+s}}=\beta$ for any $\beta\in F_q^*$.
Obviously, $\zeta^{\ell^{u-j}}$ is a primitive
$\ell^j$-th root of unity in~$F_q$.
Therefore,
$$\left(\frac{X^{\ell^{t-j}}}{\zeta^{kp^{m-s'}}}\right)^{\ell^j}-1
 =\prod_{i=0}^{\ell^j-1} \left(\frac{X^{\ell^{t-j}}}
 {\zeta^{kp^{m-s'}}}-\zeta^{i\ell^{u-j}}\right),
$$
hence
$$\left(\frac{X^{\ell^{t-j}}}{\zeta^{kp^{m-s'}}}\right)^{\ell^jp^s}-1
 =\left(\Big(\frac{X^{\ell^{t-j}}}{\zeta^{kp^{m-s'}}}\Big)^{\ell^j}-1\right)^{p^s}
 =\prod_{i=0}^{\ell^j-1}
 \left(\frac{X^{\ell^{t-j}}}{\zeta^{kp^{m-s'}}}-\zeta^{i\ell^{u-j}}\right)^{p^s}.
$$
Noting that $\zeta^{kp^{m-s'}p^{s}}=(\zeta^k)^{p^{m-s'+s}}=\zeta^k$,
we get that
\begin{equation}\label{l^t-decomp}
X^{\ell^tp^s}-\zeta^{k\ell^j}=\prod_{i=0}^{\ell^j-1}
 \left(X^{\ell^{t-j}}-\zeta^{i\ell^{u-j}+kp^{m-s'}}\right)^{p^s}.
\end{equation}
From (\ref{a-k}) and (\ref{case-ii}), we see that
$u>j$, $\ell\,|\,(p^m-1)$ and $\ell\nmid k$; hence
$\ell\mid i\ell^{u-j}$ but $\ell\nmid kp^{m-s'}$.
So $\ell\nmid(i\ell^{u-j}+kp^{m-s'})$,
hence, in the multiplicative group $F_q^*$ we have that
${\rm ord}(\zeta^{i\ell^{u-j}+kp^{m-s'}})=\ell^u$.
By Remark~\ref{irr-binomial}, all the polynomials
$$X^{\ell^{t-j}}-\zeta^{i\ell^{u-j}+kp^{m-s'}},\qquad
 i=0,1,\cdots,\ell^j-1,$$
are irreducible polynomials in $F_q[X]$, and (\ref{l^t-decomp})
is a irreducible factorization of
$X^{\ell^tp^s}-\zeta^{k\ell^j}$ in $F_q[X]$.

Replacing $X$ by $aX$, we get
$$
(aX)^{\ell^tp^s}-\zeta^{k\ell^j}=\prod_{i=0}^{\ell^j-1}
 \left((aX)^{\ell^{t-j}}-\zeta^{i\ell^{u-j}+kp^{m-s'}}\right)^{p^s}.
$$
But $a^{\ell^tp^s}\lambda=\zeta^{k\ell^j}$, i.e.
$a^{-\ell^tp^s}\zeta^{k\ell^j}=\lambda$.
We get the irreducible factorization of
$X^{\ell^tp^s}-\lambda$ in $F_q[X]$ as follows:
\begin{equation*}
X^{\ell^tp^s}-\lambda=a^{-\ell^tp^s}\prod_{i=0}^{\ell^j-1}
 \left((aX)^{\ell^{t-j}}-\zeta^{i\ell^{u-j}+kp^{m-s'}}\right)^{p^s}.
\end{equation*}
Finally, noting that
$a^{-\ell^tp^s}=\big((a^{-\ell^{t-j}})^{p^s}\big)^{\ell^{j}}$,
from the above we get the desired
irreducible factorization (\ref{lambda-lp-decomp})
of $X^{\ell^tp^s}-\lambda$ in $F_q[X]$.
\end{proof}

\begin{Remark}\label{generator}\rm
With the same notation as in Theorem \ref{thm4-1}, we can describe
the polynomial generator $g(X)$ of any $\lambda$-constacyclic code
$C$ of length $\ell^tp^s$ over $F_q$ for the two cases as follows.
\begin{itemize}
\item[{\bf(i)}:] $j=v$, then
$$ g(X)=\prod_{i=1}^{\rho} M_{\eta^{r_i}}(aX)^{e_i},\qquad
  0\le e_i\le p^s~~\forall~i=1,\cdots,\rho.$$
By the way, we show an easy subcase of this case:
if $j=v=t$, then $\zeta^{\ell^{u-t}}
=\xi^{\frac{q-1}{\ell^t}}\in F_q$
is a primitive $\ell^t$-th root of unity, hence
$X^{\ell^t}-1=\prod_{i=0}^{\ell^t-1}(X-\zeta^{i\ell^{u-t}})$;
thus the polynomial generator $g(X)$ looks simple:
\begin{equation}\label{v=t}
g(X)=\prod_{i=0}^{\ell^t-1}(X-a^{-1}\zeta^{i\ell^{u-t}})^{e_i},\quad~
 0\le e_i\le p^s~~\forall~i=0,\cdots,\ell^t-1.
\end{equation}
\item[{\bf(ii)}:] $0\le j<v\le t$, there are two subcases:
\begin{itemize}
\item[{\bf(ii.a)}:] if $\ell=2$, $t\geq2$ and $2\,\Vert\,(q-1)$, then
$$g(X)=\prod\limits_{h\in H_{t}}
 (a^{2^{t-b+1}}X^{2^{t-b+1}}-2a^{2^{t-b}}hX^{2^{t-b}}+(-1)^c)^{e_i}$$
with $0\le e_i\le p^s$ for $i=0,1,\cdots,2^{b-1}-1$.
\item[{\bf(ii.b)}:] otherwise,
$$g(X)=\prod_{i=0}^{\ell^j-1}\left(X^{\ell^{t-j}}
 -a^{-\ell^{t-j}}\zeta^{i\ell^{u-j}+kp^{m-s'}}\right)^{e_i}
$$
with $0\le e_i\le p^s$ for $i=0,1,\cdots,\ell^j-1$.
\end{itemize}\end{itemize}\end{Remark}

It is a special case for Theorem \ref{thm4-1} that
$t=v=1$, i.e. $\ell\,|\,(q-1)$ and $t=1$;
at that case, as stated in the following corollary,
there are only two $\ell p^s$-isometry classes
in $F_q^*$, and any constacyclic code of
length $\ell p^s$ over $F_q$ has a polynomial generator
with all irreducible factors being binomials.

\begin{Corollary}
Assume that $\ell$ is a prime such that
$\ell^u\Vert(q-1)$ with $u\ge 1$, $\zeta\in F_q$ is a primitive
$\ell^u$-th root of unity, and $\lambda\in F_q^*$.
Let $C$ be a $\lambda$-constacyclic code
of length $\ell p^s$ over $F_q$. Then

$\bullet$~ either $\lambda\in\langle \xi^\ell\rangle$,
$a^{\ell p^s}\lambda=1$ for an $a\in F_q$, and we have
$$ C=\left\langle \prod_{i=0}^{\ell-1}
 \big(X-a^{-1}\zeta^{i\ell^{u-1}}\big)^{e_i}\right\rangle,
 \qquad 0\le e_i\le p^s,~~\forall~i=0,1,\cdots,\ell-1;
$$

$\bullet$~ or $\lambda\notin\langle \xi^\ell\rangle$,
$a^{\ell p^s}\lambda=\zeta^k$ for an $a\in F_q^*$
and an integer $k$ coprime to $\ell p^s$, and,
taking $s'$ such that $0\le s'<m$ and $s'\equiv s~(\bmod~m)$,
we have
$$ C=\left\langle
\big(X^{\ell}-a^{-\ell}\zeta^{kp^{m-s'}}\big)^{e}\right\rangle,
 \qquad 0\le e\le p^s.
$$
\end{Corollary}

\begin{proof}
It follows from Remark \ref{generator} immediately.
We just remark that $\zeta^{\ell^{u-1}}$ is a primitive
$\ell$-th root of unity, while $\zeta^{kp^{m-s'}}$
is a primitive $\ell^u$-th root of unity.
\end{proof}

More specifically, if $\ell=2$ in the above corollary,
we reobtain the main result of \cite{Dinh11},
as stated below in our notation.

\begin{Corollary}
Assume that $2^u\Vert(q-1)$ with $u\ge 1$,
$\zeta\in F_q$ is a primitive $2^u$-th root of unity,
and $\lambda\in F_q^*$. Let $C$ be a $\lambda$-constacyclic
code of length $2 p^s$ over~$F_q$. Then

$\bullet$~ either $\lambda\in\langle \xi^2\rangle$,
$a^{2p^s}\lambda=1$ for an $a\in F_q$, and we have
$$ C=\left\langle \big(X-a^{-1}\big)^{e_0}
 \big(X+a^{-1}\big)^{e_1}\right\rangle,
 \qquad 0\le e_i\le p^s,~~\forall~i=0,1;
$$

$\bullet$~ or $\lambda\notin\langle \xi^2\rangle$,
$a^{2p^s}\lambda=\zeta^k$ for an $a\in F_q^*$
and an integer $k$ coprime to $2p^s$, and, taking
an integer $s'$ such that $0\le s'<m$ and $s'\equiv s~(\bmod~m)$,
we have
$$ C=\left\langle
\big(X^{2}-a^{-2}\zeta^{kp^{m-s'}}\big)^{e}\right\rangle,
 \qquad 0\le e\le p^s.
$$
\end{Corollary}

\begin{proof}
Just note that $\zeta^{2^{u-1}}$ is a primitive square root,
i.e. $\zeta^{2^{u-1}}=-1$.
\end{proof}

\section{Examples}

By Theorem~\ref{thm4-1},
the polynomial generators of all constacyclic codes
of length $\ell^tp^s$ over the finite field $F_{p^m}$
are easy to be established, where $\ell,p$
are different primes and $s,t$ are non-negative integers.
In this section, some examples are given to illustrate the result.

\begin{Example}
 Consider all constacyclic codes of length $6=3\cdot2$ over $F_{2^4}$. Here, $\ell=3,\,\,t=1,\,\,p=2$ and $s=1$. Let $\xi$ be a primitive $15$th root of unity in $F_{2^4}$. Since $3\,|\,(2^4-1)$, it follows that there exists primitive $3$rd root of unity in $F_{2^4}$. Therefore, $X^3-1=(X-1)(X-\xi^5)(X-\xi^{10})$. By Theorem~\ref{thm4-1}, the number of the $6$-isometry classes of $F_{2^4}^*$ is $2$.  Hence, all the constacyclic codes are divided into two parts. The polynomial generators of all  constacyclic codes are given in Table 1 and Table 2.
\end{Example}

\begin{table}[!h]\begin{center}
\begin{tabular}{|l|l|c|c|}\hline
$\lambda$& $a$ &
 $\lambda$-constacyclic codes:~
 $0\leq j_0,j_1,j_2\leq2$  & sizes\\ \hline
 1  & $1$ &
 $\langle(X-1)^{j_0}(X-\xi^5)^{j_1}(X-\xi^{10})^{j_2}\rangle$&
 $16^{6-j_o-j_1-j_2}$\\
$\xi^3$ & $\xi^7$ &
$\langle(\xi^7X-1)^{j_0}(\xi^7X-\xi^5)^{j_1}(\xi^7X-\xi^{10})^{j_2}\rangle$ & $16^{6-j_o-j_1-j_2}$\\
$\xi^6$ & $\xi^4$ &
$\langle(\xi^4 X-1)^{j_0}(\xi^4 X-\xi^5)^{j_1}(\xi^4 X-\xi^{10})^{j_2}\rangle$ & $16^{6-j_o-j_1-j_2}$\\
$\xi^9$   & $\xi$ &
$\langle(\xi X-1)^{j_0}(\xi X-\xi^5)^{j_1}(\xi X-\xi^{10})^{j_2}\rangle$ & $16^{6-j_o-j_1-j_2}$\\
$\xi^{12}$ & $\xi^3$ &
$\langle(\xi^3X-1)^{j_0}(\xi^3X-\xi^5)^{j_1}(\xi^3X-\xi^{10})^{j_2}\rangle$ & $16^{6-j_o-j_1-j_2}$\\
\hline
\end{tabular}
\caption{$\lambda$-constacyclic codes of length $6$ over $F_{2^4}$, $\lambda\cong_61$,  $ a^6\lambda=1$ }
\end{center}\end{table}

\smallskip
\begin{table}[!h]\begin{center}
\begin{tabular}{|l|l|l|c|c|}\hline
$\lambda$ & $k$  & $a$ &
$\lambda$-constacyclic codes:~ $0\leq j\leq2$ & sizes \\ \hline
$\xi$       & $5$  & $\xi^4$ & $\langle(X^3-\xi^8)^j\rangle$           & $16^{6-3j}$\\
$\xi^{4}$   & $5$  & $\xi^6$ & $\langle(X^3-\xi^2)^j\rangle$     & $16^{6-3j}$\\
$\xi^{7}$   & $5$  & $\xi^8$ & $\langle(X^3-\xi^{11})^j\rangle$  & $16^{6-3j}$\\
$\xi^{10}$  & $5$  & $\xi^5$ & $\langle(X^3-\xi^5)^j\rangle$       & $16^{6-3j}$\\
$\xi^{13}$  & $5$  & $\xi^2$ & $\langle(X^3-\xi^{14})^j\rangle$  & $16^{6-3j}$\\
$\xi^2$     & $1$  & $\xi^3$ & $\langle(X^3-\xi)^j\rangle$       & $16^{6-3j}$\\
$\xi^{5}$   & $1$  & $1$ & $\langle(X^3-\xi^{10})^j\rangle$      & $16^{6-3j}$\\
$\xi^{8}$   & $1$  & $\xi^2$ & $\langle(X^3-\xi^{4})^j\rangle$   & $16^{6-3j}$\\
$\xi^{11}$  & $1$  & $\xi^4$ & $\langle(X^3-\xi^{13})^j\rangle$  & $16^{6-3j}$\\
$\xi^{14}$  & $1$  & $\xi$ & $\langle(X^3-\xi^{7})^j\rangle$     & $16^{6-3j}$\\
\hline
\end{tabular}
\caption{$\lambda$-constacyclic codes of length $6$
over $F_{2^4}$, $\lambda\cong_6\xi^5$,  $ a^6\lambda=\xi^{5k}$ }
\end{center}\end{table}

\begin{table}[!h]\begin{center}
\begin{tabular}{|l|l|c|c|}\hline
$\lambda$& $a$ &
 $\lambda$-constacyclic codes:~
 $0\leq i,j,k\leq25$  & sizes\\ \hline
 $1$         & $1$        & $\langle(X-1)^i g(X)^jh(X)^k\rangle$                         & $25^{175-i-3j-3k}$\\
 $\xi$       & $\xi^{17}$ & $\langle(\xi^{17}X-1)^i g(\xi^{17}X)^jh(\xi^{17}X)^k\rangle$ & $25^{175-i-3j-3k}$\\
 $\xi^2$     & $\xi^{10}$ & $\langle(\xi^{10}X-1)^i g(\xi^{10}X)^jh(\xi^{10}X)^k\rangle$ & $25^{175-i-3j-3k}$\\
 $\xi^3$     & $\xi^3$    & $\langle(\xi^3X-1)^i g(\xi^3X)^jh(\xi^3X)^k\rangle$          & $25^{175-i-3j-3k}$\\
 $\xi^4$     & $\xi^{20}$ & $\langle(\xi^{20}X-1)^i g(\xi^{20}X)^jh(\xi^{20}X)^k\rangle$ & $25^{175-i-3j-3k}$\\
 $\xi^5$     & $\xi^{13}$ & $\langle(\xi^{13}X-1)^i g(\xi^{13}X)^jh(\xi^{13}X)^k\rangle$ & $25^{175-i-3j-3k}$\\
 $\xi^6$     & $\xi^6$    & $\langle(\xi^6X-1)^i g(\xi^6X)^jh(\xi^6X)^k\rangle$          & $25^{175-i-3j-3k}$\\
 $\xi^7$     & $\xi^{23}$ & $\langle(\xi^{23}X-1)^i g(\xi^{23}X)^jh(\xi^{23}X)^k\rangle$ & $25^{175-i-3j-3k}$\\
 $\xi^8$     & $\xi^{16}$ & $\langle(\xi^{16}X-1)^i g(\xi^{16}X)^jh(\xi^{16}X)^k\rangle$ & $25^{175-i-3j-3k}$\\
 $\xi^9$     & $\xi^9$    & $\langle(\xi^9X-1)^i g(\xi^9X)^jh(\xi^9X)^k\rangle$          & $25^{175-i-3j-3k}$\\
 $\xi^{10}$  & $\xi^2$    & $\langle(\xi^2X-1)^i g(\xi^2X)^jh(\xi^2X)^k\rangle$          & $25^{175-i-3j-3k}$\\
 $\xi^{11}$  & $\xi^{19}$ & $\langle(\xi^{19}X-1)^i g(\xi^{19}X)^jh(\xi^{19}X)^k\rangle$ & $25^{175-i-3j-3k}$\\
 $\xi^{12}$  & $\xi^{12}$ & $\langle(\xi^{12}X-1)^i g(\xi^{12}X)^jh(\xi^{12}X)^k\rangle$ & $25^{175-i-3j-3k}$\\
 $\xi^{13}$  & $\xi^5$    & $\langle(\xi^5X-1)^i g(\xi^5X)^jh(\xi^5X)^k\rangle$          & $25^{175-i-3j-3k}$\\
 $\xi^{14}$  & $\xi^{22}$ & $\langle(\xi^{22}X-1)^i g(\xi^{22}X)^jh(\xi^{22}X)^k\rangle$ & $25^{175-i-3j-3k}$\\
 $\xi^{15}$  & $\xi^{15}$ & $\langle(\xi^{15}X-1)^i g(\xi^{15}X)^jh(\xi^{15}X)^k\rangle$ & $25^{175-i-3j-3k}$\\
 $\xi^{16}$  & $\xi^8$    & $\langle(\xi^8X-1)^i g(\xi^8X)^jh(\xi^8X)^k\rangle$          & $25^{175-i-3j-3k}$\\
 $\xi^{17}$  & $\xi$      & $\langle(\xi X-1)^i g(\xi X)^jh(\xi X)^k\rangle$             & $25^{175-i-3j-3k}$\\
 $\xi^{18}$  & $\xi^{18}$ & $\langle(\xi^{18}X-1)^i g(\xi^{18}X)^jh(\xi^{18}X)^k\rangle$ & $25^{175-i-3j-3k}$\\
 $\xi^{19}$  & $\xi^{11}$ & $\langle(\xi^{11}X-1)^i g(\xi^{11}X)^jh(\xi^{11}X)^k\rangle$ & $25^{175-i-3j-3k}$\\
 $\xi^{20}$  & $\xi^4$    & $\langle(\xi^4X-1)^i g(\xi^4X)^jh(\xi^4X)^k\rangle$          & $25^{175-i-3j-3k}$\\
 $\xi^{21}$  & $\xi^{21}$ & $\langle(\xi^{21}X-1)^i g(\xi^{21}X)^jh(\xi^{21}X)^k\rangle$ & $25^{175-i-3j-3k}$\\
 $\xi^{22}$  & $\xi^{14}$ & $\langle(\xi^{14}X-1)^i g(\xi^{14}X)^jh(\xi^{14}X)^k\rangle$ & $25^{175-i-3j-3k}$\\
 $\xi^{23}$  & $\xi^7$    & $\langle(\xi^7X-1)^i g(\xi^7X)^jh(\xi^7X)^k\rangle$          & $25^{175-i-3j-3k}$\\
\hline
\end{tabular}
\caption{$\lambda$-constacyclic codes of length $175$ over $F_{5^2}$, $\lambda\cong_{175}1$,  $ a^{175}\lambda=1$ }
\end{center}\end{table}

\begin{Example}
Consider all constacyclic codes of length $175=7\cdot5^2$ over $F_{5^2}$.
Here, $\ell=7,\,\,t=1,\,\,p=5$ and $s=2$. Let $\xi$ be a primitive $24$th root of unity in $F_{5^2}$.
Since $\gcd(175, 5^2-1)=1$, by Corollary~\ref{cor3-1}, all the constacyclic codes of length~$175$
are isometric to the cyclic codes of length~$175$.
By \cite{GAP}, it follows that
 $X^7-1=(X-1)(X^3+\xi X^2+\xi^{17}X-1)(x^3+\xi^5X^2+\xi^{13}X-1 )$ is the factorization of $X^7-1$
into irreducible factors over $F_{5^2}$.
Let $g(X)=X^3+\xi X^2+\xi^{17}X-1$ and $h(X)=x^3+\xi^5X^2+\xi^{13}X-1$.
The polynomial generators of constacyclic codes are given in Table 3.
\end{Example}

\begin{Example}
 Consider all constacyclic codes of length $20=2^2\cdot5$ over $F_{5^2}$. Here, $\ell=2,\,\,t=2,\,\,p=5$ and $s=1$. Let $\xi$ be a primitive $24$th root of unity in $F_{5^2}$.  Since $4\,|\,(5^2-1)$, it follows that there exists a primitive $4$th root of identity in $F_{5^2}$. Therefore, $X^4-1=(X-1)(X-\xi^6)(X-\xi^{12})(X-\xi^{18})$. By Theorem~\ref{thm4-1}, the number of the $20$-isometry classes of $F_{2^4}^*$ is $3$. The polynomial generators of constacyclic codes are given in Table 4-6.
\end{Example}

\begin{table}[h]\begin{center}
{\footnotesize
\begin{tabular}{|l|l|c|c|}\hline
$\lambda$& $a$ &
 $\lambda$-constacyclic codes:~
 $0\leq j_0,j_1,j_2,j_3\leq5$  & sizes\\ \hline
  1         & $\xi^6$ & $\langle(\xi^6X-1)^{j_0}(\xi^6X-\xi^6)^{j_1}(\xi^6X-\xi^{12})^{j_2}(\xi^6X-\xi^{18})^{j_3}\rangle$ & $25^{20-j_o-j_1-j_2-j_3}$\\
  $\xi^4$   & $\xi$   & $\langle(\xi X-1)^{j_0}(\xi X-\xi^6)^{j_1}(\xi X-\xi^{12})^{j_2}(\xi X-\xi^{18})^{j_3}\rangle$ & $25^{20-j_o-j_1-j_2-j_3}$\\
 $\xi^8$    & $\xi^2$ & $\langle(\xi^2X-1)^{j_0}(\xi^2X-\xi^6)^{j_1}(\xi^2X-\xi^{12})^{j_2}(\xi^2X-\xi^{18})^{j_3}\rangle$ & $25^{20-j_o-j_1-j_2-j_3}$\\
 $\xi^{12}$ & $\xi^3$ & $\langle(\xi^3X-1)^{j_0}(\xi^3X-\xi^6)^{j_1}(\xi^3X-\xi^{12})^{j_2}(\xi^3X-\xi^{18})^{j_3}\rangle$ & $25^{20-j_o-j_1-j_2-j_3}$\\
 $\xi^{16}$ & $\xi^4$ & $\langle(\xi^4X-1)^{j_0}(\xi^4X-\xi^6)^{j_1}(\xi^4X-\xi^{12})^{j_2}(\xi^4X-\xi^{18})^{j_3}\rangle$ & $25^{20-j_o-j_1-j_2-j_3}$\\
 $\xi^{20}$ & $\xi^5$ & $\langle(\xi^5X-1)^{j_0}(\xi^5X-\xi^6)^{j_1}(\xi^5X-\xi^{12})^{j_2}(\xi^5X-\xi^{18})^{j_3}\rangle$ & $25^{20-j_o-j_1-j_2-j_3}$\\
\hline
\end{tabular}}
\caption{$\lambda$-constacyclic codes of length $20$ over $F_{5^2}$, $\lambda\cong_{20}1$,  $ a^{20}\lambda=1$ }

\end{center}\end{table}

\smallskip
\begin{table}[h]\begin{center}
\begin{tabular}{|l|l|l|c|c|}\hline
$\lambda$ & $k$  & $a$ &
$\lambda$-constacyclic codes:~ $0\leq j\leq5$ & sizes \\ \hline
 $\xi$      &  $3$   & $\xi^4$    & $\langle(X^4-\xi^5)^j\rangle$    & $25^{20-4j}$\\
 $\xi^5$    &  $3$   & $\xi^{23}$        & $\langle(X^4-\xi^5)^j\rangle$    & $25^{20-4j}$\\
 $\xi^9$    &  $3$   & $\xi^6$      & $\langle(X^4-\xi^{21})^j\rangle$  & $25^{20-4j}$\\
 $\xi^{13}$ &  $3$   & $\xi$    & $\langle(X^4-\xi^{14})^j\rangle$ & $25^{20-4j}$\\
 $\xi^{17}$ &  $3$   & $\xi^2$    & $\langle(X^4-\xi^{13})^j\rangle$ & $25^{20-4j}$\\
 $\xi^{21}$ &  $3$   & $\xi^3$    & $\langle(X^4-\xi^9)^j\rangle$    & $25^{20-4j}$\\
 $\xi^3$    &  $1$   & $1$    & $\langle(X^4-\xi^{15})^j\rangle$ & $25^{20-4j}$\\
 $\xi^7$    &  $1$   & $\xi$    & $\langle(X^4-\xi^{11})^j\rangle$ & $25^{20-4j}$\\
 $\xi^{11}$ &  $1$   & $\xi^2$ & $\langle(X^4-\xi^{7})^j\rangle$  & $25^{20-4j}$\\
 $\xi^{15}$ &  $1$   & $\xi^3$        & $\langle( X^4-\xi^{3})^j\rangle$ & $25^{20-4j}$\\
 $\xi^{19}$ &  $1$   & $\xi^4$      & $\langle(X^4-\xi^{23})^j\rangle$ & $25^{20-4j}$\\
 $\xi^{23}$ &  $1$   & $\xi^5$    & $\langle(X^4-\xi^{19})^j\rangle$ & $25^{20-4j}$\\
\hline
\end{tabular}
\caption{$\lambda$-constacyclic codes of length $20$ over $F_{5^2}$, $\lambda\cong_{20}\xi^3$,
$a^{20}\lambda=\xi^{3k}$}
\end{center}\end{table}

\smallskip
\begin{table}[!h]\begin{center}
\begin{tabular}{|l|l|l|c|c|}\hline
$\lambda$ & $k$  & $a$ &
$\lambda$-constacyclic codes:~ $0\leq j_0,j_1\leq5$ & sizes \\ \hline
 $\xi^2$    & $1$  & $\xi^{23}$    &  $\langle(X^2-\xi^{5})^{j_0}(X^2+\xi^{5})^{j_1}\rangle$ &  $25^{20-2j_o-2j_1}$\\
 $\xi^6$    & $1$  & $\xi^{6}$ &  $\langle(X^2-\xi^{15})^{j_0}(X^2+\xi^{15})^{j_1}\rangle$       &  $25^{20-2j_o-2j_1}$\\
 $\xi^{10}$ & $1$  & $\xi$        &  $\langle(X^2-\xi)^{j_0}(X^2+\xi)^{j_1}\rangle$           &  $25^{20-2j_o-2j_1}$\\
 $\xi^{14}$ & $1$  & $\xi^2$      &  $\langle(X^2-\xi^{23})^{j_0}(X^2+\xi^{23})^{j_1}\rangle$ &  $25^{20-2j_o-2j_1}$\\
 $\xi^{18}$ & $1$  & $\xi^3$    &  $\langle(X^2-\xi^{18})^{j_0}(X^2+\xi^{18})^{j_1}\rangle$       & $25^{20-2j_o-2j_1}$\\
 $\xi^{22}$ & $1$  & $\xi^4$    &  $\langle(X^2-\xi^{11})^{j_0}(X^2+\xi^{11})^{j_1}\rangle$       & $25^{20-2j_o-2j_1}$\\
\hline
\end{tabular}
\caption{$\lambda$-constacyclic codes of length $20$ over $F_{5^2}$, $\lambda\cong_{20}\xi^6$,  $ a^{20}\lambda=\xi^{6k}$}
\end{center}\end{table}

\newpage

\noindent{\bf Acknowledgements}

This work was supported by NSFC, Grant No.~11171370, and
Research Funds of CCNU, Grant No.~11A02014.
The authors would like to thank the anonymous referees
for their many helpful comments.

\end{document}